\documentclass[a4paper]{llncs}

\usepackage[utf8]{inputenc}
\usepackage{amsmath}
\usepackage{amssymb}
\usepackage{color}
\usepackage{epsfig}
\usepackage{subfigure}

\usepackage{bm}
\usepackage{pgf}
\usepackage{tikz}
\usetikzlibrary{trees}
\usetikzlibrary{arrows,shapes,snakes,automata,backgrounds,petri}

\usepackage{xspace}
\usepackage{listings}
\usepackage{url}

\newcommand{\ourtitle}{Towards Bounded Infeasible Code Detection}

\usepackage{macros}

\usepackage[ruled,lined,longend]{algorithm2e}

\definecolor{darkred}{rgb}{0.5,0,0}
\definecolor{darkblue}{rgb}{0,0,0.5}
\definecolor{darkgreen}{rgb}{0,0.5,0}

\usepackage[colorlinks]{hyperref}
   
    \hypersetup{colorlinks
   ,linkcolor=darkred
   ,filecolor=darkgreen
   ,urlcolor=darkred
   ,citecolor=darkblue   
   ,pdftitle={\ourtitle}  
   ,pdfauthor={Stephan Arlt, J\"{u}rgen Christ, Jochen Hoenicke, Martin
   Sch\"{a}f} 
   ,pdfkeywords={feasible path cover }{error detection }{infeasible code}
   ,pdfsubject={error detection }
   }

\begin{document}

\title{\ourtitle}

\def\arlt{\protect\href{http://swt.informatik.uni-freiburg.de/~arlt}{Stephan
Arlt}}

\def\christ{\protect\href{http://swt.informatik.uni-freiburg.de/~christj}{J\"{u}rgen
Christ}}

\def\schaef{\protect\href{http://iist.unu.edu/people/schaef}{Martin 
  Sch\"af}}
\def\hoenicke{\protect\href{http://swt.informatik.uni-freiburg.de/~hoenicke}{Jochen Hoenicke}}

\author{\christ\inst{1}, \hoenicke\inst{1}, and
\schaef\inst{2}}

\institute{Albert-Ludwig University, Freiburg
\and
United Nations University, IIST, Macau
  }

\maketitle
\begin{abstract}

A first step towards more reliable software is to execute each statement and
each control-flow path in a method once. In this paper, we present a formal
method to automatically compute test cases for this purpose based on the idea of
a bounded infeasible code detection.
The method first unwinds all loops in a program finitely often and then encodes
all feasible executions of the loop-free programs in a logical formula. Helper
variables are introduced such that a theorem prover can reconstruct the
control-flow path of a feasible execution from a satisfying valuation of this
formula. Based on this formula, we present one algorithm that computes a
feasible path cover and one algorithm that computes a feasible statement cover.
We show that the algorithms are complete for loop-free programs and that they
can be implemented efficiently.
We further provide a sound algorithm to compute procedure summaries which makes
the method scalable to larger programs.
\end{abstract}

\section{Introduction}\label{sec:introduction}

Using static analysis to find feasible executions of a program that pass a
particular subset of program statements is an interesting problem. Even
though in general not decidable, there is ongoing research effort to develop
algorithms and tools that are able to solve this problems for a reasonable large
number of cases. Such tools can be used, e.g., to automatically generate test
cases that cover large portions of a programs source code and trigger rare
behavior, or to identify program fragments for which no suitable test case can
be found. The later case sometimes is referred to as \emph{infeasible code
detection}~\cite{vstte12}. Code is considered to be infeasible if no terminating
execution can be found for it. Infeasible code can be seen as a superset of
unreachable code as there might be executions reaching a piece of infeasible
code which, however, fail during their later execution.

In particular, a counterexample for the infeasibility of a piece of code is a
terminating execution that executes this code. That is, finding a set of test
cases that cover all statements in a program is equivalent to proving the
absence of infeasible code. Existing approaches to detect infeasible code do not
yet exploit the fact that counterexamples for infeasibility might constitute
feasible test cases.

In this paper, we discuss a bounded approach towards infeasible code detection
that generates test cases that cover all statements which have feasible
executions within a given (bounded) number of loop unwindings.
The interesting aspect of bounded infeasible code detection over existing
(unbounded) approaches is that counterexamples for infeasibility are likely to
represent actual executions of the program, as compared to the unbounded case,
where these counterexamples might be introduced by the
necessary over-approximation of the feasible executions.

The paper proposes two novel ideas: the concept of \emph{reachability
verification condition}, which is a formula representation of the program which,
similar the weakest-liberal precondition or strongest postcondition,
models all feasible executions of a program. But in contrast to existing
concepts, a satisfying assignment to the reachability verification condition
can directly be mapped to an execution of the program from source to sink. For
example, a valuation of \emph{wlp} can represent a feasible execution starting
from any point in a program, but this does not yet imply that this point is
actually reachable from the initial states of the program. Certainly there are
ways to encode the desired property using \emph{wlp}, or \emph{sp} by adding
helper variables to the program (see, e.g.,~\cite{doomedjournal,vstte12}),
however, we claim that the proposed reachability verification condition provides
a better formal basis to show the absence of infeasible code, as it, e.g., can
make better use of the theorem prover stack which results in a more efficient
and scalable solution. We suggest two algorithms to compute feasible executions
of a program based on the reachability verification condition. One uses
so-called \emph{blocking clauses} to prevent the theorem prover from exercising
the same path twice, the other algorithm uses \emph{enabling clauses} to urge
the theorem prover to consider a solution that passes program fragments that have
not been accessed before. Both algorithms return a set of feasible executions in
the bounded program. Further, both algorithms guarantee that any statement not
executed by these test cases is infeasible within the given bounds. We do a
preliminary evaluation of our algorithms against existing algorithms to detect
infeasible code.

Based on the reachability verification condition, as a second novelty, we
propose a technique to compute procedure summaries for bounded infeasible code
detection. As the presented algorithms return a set of feasible executions, we
can extract pairs of input and output values for each execution to construct
procedure summaries. The summaries are a strict under-approximation of the
possible executions of the summarized procedure. Therefore, the computed
summaries are sound to show the presence of feasible executions, but unsound to
show their absence. To overcome this gap, we suggest an on-demand computation of
summaries if no feasible execution can be found with the given summary. Within
the scope of this paper, we do not evaluate the concept of summaries as more
implementation effort is required until viable results can be presented.

In Section~\ref{sec:loopunwinding} we explain how we address the problem of
computing the weakest-liberal preconditions for general programs. In
Section~\ref{sec:vc} we show how a feasible execution that visits certain blocks
can efficiently be expressed as a formula and introduce the concept of
reachability verification condition. In Section~\ref{sec:algorithm} we present
two different algorithms to address the problem of generating test cases with
optimal coverage. In Section~\ref{sec:procsum} we show how procedure summaries
can be computed with our test case generation algorithm. We present an
experimental evaluation of our algorithms in Section~\ref{sec:experiments}.

\section{Preliminaries}\label{sec:preliminaries}

For simplicity, we consider only simple
unstructured programs written in the language given in Figure~\ref{fig:language}.
\begin{figure}[htdp]
\begin{center}
\begin{align*}
  \textit{Program} ::= & \; \textit{Procedure}^+ \\
  \textit{Procedure} ::= & \; \textbf{proc} \, \textit{ProcName}
  \textbf{(}\textit{VarId}^*\textbf{)} \, [ \, \textbf{returns} \, \textit{VarId} \, ] \, \textbf{\{} \,
    \textit{Block}^+ \, \textbf{\}} \\
  \textit{Block}  ::= & \; \textit{label} : \; \textit{Stmt}^* \,
  [ \, \textbf{goto} \,  \textit{label}^+\textbf{;} \, ]\\ \textit{Stmt} ::= & \; \textit{VarId}
  := \textit{Expr}\textbf{;} \; | \;
  \textbf{assume} \, \textit{Expr}\textbf{;} \\
  &  \mid  \textit{VarId} := \textbf{call}\,
  \textit{ProcName}\textbf{(}\textit{Expr}^*\textbf{)}\textbf{;}
 \end{align*}
\end{center}
\caption{Simple (unstructured) Language \label{fig:language}}
\end{figure}

Expressions are sorted first order logic terms of appropriate sort.
The expression after an \textbf{assume} statement have
Boolean sort.
A program is given by a set of \textit{Procedures} each with a unique name.
The special procedure named ``main'' is the entry point of a program.
Every procedure contains at least one block of code.
A block consists of a
label, a (possibly empty) sequence of statements, and non-deterministic
$\goto$ statement that lists transitions to successor blocks. 
The $\goto$ statement is omitted for the blocks that have no successors.
A statement can either be an assignment of a term to a
variable, an assumption, or a procedure call.
A call to a procedure is indicated by the \textbf{call} keyword followed by
the name of the procedure to call, and the (possible empty) list of arguments.
A procedure can return a value by writing into the variable mentioned in the
\textbf{returns} declaration.
If this declaration is omitted, the procedure cannot return a value.
\begin{figure}[htbp]
\begin{verbatim}
proc foo(x, y) returns z {
  l0: 
      goto l1, l2;
  l1: assume y > 0;
      z := x + y;
      goto l3;
  l2: assume y <= 0;
      z := x - y;
      goto l3;
  l3:
}

proc main() {
  l0: r := call foo(0, 1);
}
\end{verbatim}
  \caption{\label{fig:languageexample}Example of our Simple Language}
\end{figure}
If the conditional of an assumption evaluates to $\false$, the
execution blocks.
Figure~\ref{fig:languageexample} shows a small example of our simple language.

We assume that every procedure contains a unique \emph{initial block} $Block_0$
and a unique \emph{final block} that has no successor.  
A procedure \emph{terminates} if it reaches the
end of the final block.  A program \emph{terminates} if the
``main'' procedure terminates. We further assume
the directed graph which is given by the transitions between the
blocks is reducible.

The presented language is simple but yet expressive enough to encode high level
programming languages such as \texttt{C}~\cite{Cohen:2009:VPS:1616077.1616080}.
In this paper we do not address the problems that can arise during this
translation and refer to related work instead.

The weakest-liberal precondition~\cite{nla.cat-vn2681671,barnett2005} semantics
of our language is defined in the standard way:
\begin{center}
\begin{tabular}{ c | c }
  $\mathit{\st}$ &   $\wlp ( \mathit{\st}, Q)$ \\
  \hline
  $\textbf{assume}\;  E $ & $E \implies Q$ \\
  $ \textit{VarId} := \textit{Expr} $ & $Q[\textit{VarId}/ \textit{Expr}]$ \\    
  $S;T$ & $\wlp(S,\wlp(T,Q))$ \\
 \end{tabular}
\end{center}
A sequence of statements $\st$ in our language has a \emph{feasible execution} if and only if
there exists an initial valuation $\val$ of the program variables, such that in the execution of
$\st$ all \textbf{assume} statements are satisfied.
\begin{theorem}
A sequence of statements $\st$ has a feasible execution if and only if there exists a
valuation $\val$ of the program variables, such that $\val \not\models
\wlp(\st, \false)$.
\end{theorem}
Hence, the initial state of a feasible execution of $\st$ can be derived
from a counterexample to the formula representation of the weak-liberal
precondition $\wlp(\st, \false)$.

A path in a program is a sequence of blocks $\pi = \textit{Block}_0
\ldots \textit{Block}_n$ such that there is a transition from any
$\textit{Block}_i$ to $\textit{Block}_{i+1}$ for $0\leq i<n$.
We extend the definition of feasible executions from statements to paths 
by concatenating the statements of each block.
We say that a path
$\pi$ is a \emph{complete path} if it starts in the initial block and ends in 
the final block. In the following, we always refer to complete paths unless
explicitly stated differently. A path is \emph{feasible}, if there exists a
feasible execution for that path.
\begin{theorem}
Given a path $\pi = \textit{Block}_0 \ldots \textit{Block}_n $ in a program
$\prg$ where $\st_i$ represents the statements of $\textit{Block}_i$. The path
$\pi$ is called feasible, if and only if there exists a valuation $\val$ of the
program variables, such that $\val \not\models \wlp(\st_0; \ldots; \st_n, \false)$.
\end{theorem}

Note that our simple language does \emph{not} support assertions. For the
weakest liberal precondition, assertions are treated in the same way as
assumptions. That is, we might render a path infeasible because it's execution
fails, but still this path might be executable. As our goal is to execute all
possible control-flow paths, we encode assertions as conditional choice. This
allows us later on to check if there exist test cases that violate an assertion.


\section{Program Transformation}\label{sec:loopunwinding}

As the weakest-liberal precondition cannot be computed for programs with loops
in the general case, an abstraction is needed. Depending on the purpose of the
analysis, different information about the possible executions of the program
has to be preserved to retain \emph{soundness}. E.g., when proving partial
correctness~\cite{Barnett06boogie,barnett2005} of a program, the set of all
executions that fail has to be preserved (or might be
over-approximated), while terminating or blocking executions might be
omitted or added.

For our purpose of identifying a set of executions containing all feasible
statements, such an abstraction, which over-approximates the executions of
a program is not suitable as we might report executions which do not exist in
the original program. Instead we need a loop unwinding which does not add any
(feasible) executions. 

\paragraph{Loop Unwinding.}
Our loop unwinding technique is sketched in Figure~\ref{fig:unwinding}. As we
assume (w.l.o.g) that the control-flow graph of our input program is reducible,
we can identify one unique entry point for each loop, the \emph{loop header}
$B_h$, and a \emph{loop body} $B$. The loop header contains only a transition
to the loop body and the \emph{loop exit} $B_e$. 
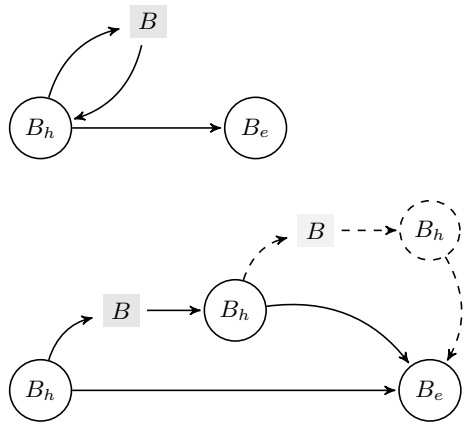
\begin{figure}
\centering
\begin{minipage}{.73\linewidth}
\begin{tikzpicture}[->,>=stealth',shorten >=1pt,auto,node distance=2cm,
                    semithick]
  \tikzstyle{every state}=[fill=none,draw,text=black]
  \tikzstyle{lbody}=[line width=2mm,join=round,fill=black!10]  
  
  \node[state]         (A)                    {$B_h$};
  \node[lbody]         (B) [above right of=A]   
  {$B$};
  \node[state]         (C) [below right of=B] {$B_e$};
   
  \path (A) edge  [bend left] node {}  (B)     
  	         edge	node {} (C)       
        (B) edge [bend left]  node {} (A);
                
\end{tikzpicture}
\end{minipage}

\vspace{0.5cm}
\begin{minipage}{.73\linewidth}
\begin{tikzpicture}[->,>=stealth',shorten >=1pt,auto,node distance=1.5cm,
                    semithick]
  \tikzstyle{every state}=[fill=none,draw,text=black]
  \tikzstyle{lbody}=[line width=2mm,join=round,fill=black!10]  
  \tikzstyle{dottedstate}=[circle,fill=none,draw,text=black,dashed]
  \tikzstyle{dottedlbody}=[line width=2mm,join=round,fill=black!5, dashed]  
  
  \node[state]         (A)                    {$B_h$};
  \node[lbody]         (B) [above right of=A] {$B$};   
  \node[state]         (C) [right of=B] {$B_h$};
  \node[dottedlbody]         (D) [above right of=C] {$B$};
  \node[dottedstate]         (E) [right of=D] {$B_h$};  

  \node[state]         (F) [below right of=D, right of=C] {$B_e$};
   
  \path (A) edge  [bend left] node {}  (B)     
  	         edge	node {} (F)       
        (B) edge node {} (C)
        (C) edge [bend left] node {} (F);
  
  \path[dashed] (C) edge [bend left]  node {}  (D)
  		(D) edge   node {}  (E)
  		(E) edge [bend left]  node {}  (F);
                
\end{tikzpicture}
\end{minipage}
\caption[Loop abstraction]{Finite loop unwinding} \label{fig:unwinding}
\end{figure}
We can now unwind the loop once by simply redirecting the target of the
back-edge that goes from $B$ to $B_h$, to $B_e$ (and thus transforming the loop
into an if-then-else). 

To unwind the loop $k$-times, for each unwinding, we have to create a copy of
$B$ and $B_h$, and redirect the outgoing edge of the $B$ introduced in the
previous unwinding to the newly introduced $B_h$. That is, the loop is
transformed to an if-then-else tree of depth $k$.
%
%
This abstraction is limited to finding executions that reach statements within
less than $(k+1)$ loop iterations, however, as the abstraction never adds a
feasible execution, we have the guarantee that this execution really exists.
\begin{lemma}\label{thrm:soundness}
Given a program $\prg$ and a program $\prg'$ which is generated from $\prg$ by
$k$-times loop unwinding. Any feasible execution of $\prg'$ is also a feasible
execution of $\prg$.
\end{lemma}

\paragraph{Procedure Calls.} Procedure calls are another problem when computing
the weakest (liberal) precondition. First, they can introduce looping control
flow via recursion, and second, inlining each procedure call might dramatically
increase the size of the program that has to be considered. For recursive
procedure calls, we can apply the same loop unwinding used for normal loops.

To inline a procedure, we split the block at the location of the procedure
call in two blocks and add all blocks of the body of the called procedure in
between (and rename variables and labels if necessary). Then, we add additional assignments to
map the parameters of the called procedure to the arguments used in the
procedure call and the variable carrying the return value of the procedure to
those receiving it in the calling procedure.

If inlining all procedure calls is not feasible due to the size of the program,
the call has to be replaced by a summary of the procedure body instead. We
propose a technique that retains the soundness from Lemma~\ref{thrm:soundness}
later on in Section~\ref{sec:procsum}.

\paragraph{Single Static Assignment.} For the resulting loop-free program, we
perform a single static assignment transformation~\cite{Cytron:1991:ECS:115372.115320}
which introduces auxiliary variables to ensure that each program variable is
assigned at most once on each execution path
\cite{Flanagan:2001:AEE:360204.360220}. For convenience we use the following
notation: given a program variable $v$, the single static assignment
transformation transforms an assignment $v := v + 1$ into $v_{i+1} := v_{i} + 1$,
where $v_{i+1}$ and $v_{i}$ are auxiliary variable (and the index represents the
incarnation of $v$). In the resulting program, each variable is written at most
once. Hence, we can replace all assignments by assumptions without altering
the feasible executions of the program. In that sense, the transformed program
is passive as it does not change the values of variables. As single static assignment is used
frequently in verification, we refer to the related work for more details (e.g.,
\cite{Flanagan:2001:AEE:360204.360220,Leino:2005:EWP:1066417.1066421,barnett2005}).

\section{Reachability Verification Condition}\label{sec:vc}

This section explains how to find a formula $\vc$ (the reachability
verification condition) such that every satisfying valuation $\val$ 
corresponds to a terminating execution of the program.  Moreover, it is
possible to determine from the valuation, which blocks of the program
were reached by this execution.  For this purpose $\vc$ contains an
auxiliary variable $\rvar_i$ for each block that is true if the block is
visited by the execution.  From such an execution we can derive a test
case by looking at the initial valuation of the variables.

A test case of a program can be found using the weakest (liberal)
precondition.  If a state satisfies the weakest precondition
$wp(S,\true)$ of a program $S$ it will produce a non-failing run.
However, it may still block in an \texttt{assume} statement.  Since we
desire to find non-blocking test cases we follow~\cite{vstte12} and
use the weakest liberal precondition of false.  A state satisfies
$\wlp(S, \false)$ if and only if it does not terminate.  Hence we
can use $\lnot \wlp(S,\false)$ to find terminating runs of $S$.

For a loop-free program, computing the weakest (liberal) precondition
is straight forward and has been discussed in many previous articles
(e.g.,
\cite{barnett2005,Leino:2005:EWP:1066417.1066421,Flanagan:2002:ESC:543552.512558,Grigore:2009:SPU:1557898.1557904}).
To avoid exponential explosion of the formula size, for each block
\[
\textit{Block}_i ::= \, i : S_i; \goto Succ_i
\] 
we introduce an auxiliary variable $B_i$ that represents the formula
$\lnot \wlp(Block_i, \false)$, where $Block_i$ is the program fragment
starting at label $i$ and continuing to the termination point of the
program.  These variables can be defined as
\begin{align*}
\WLP : \bigwedge_{0\leq i< n}
B_i &\equiv \lnot \wlp\biggl(S_i, \bigwedge_{j \in Succ_i} \lnot B_j\biggr) \\
{}\land B_n &\equiv \lnot\wlp(S_n, \false). 
\end{align*}
Introducing the auxiliary variables avoids copying the $\wlp$ of the successor 
blocks. 
If we are interested in a terminating execution that starts in the 
initial location $0$, we can
find a satisfying valuation for
$$ \WLP \land  B_0$$
\begin{lemma}\label{lemma:wlp}
There is a satisfying valuation $\val$ for the formula $\WLP$ with
$\val(B_i)= \true$ if and only if there is a 
terminating execution for the program fragment starting at the block 
$Block_i$.
\end{lemma}
Proof is given in \cite{vstte12}.

Thus a satisfying valuation $\val$ of $\WLP \land B_0$
corresponds to a terminating execution of the whole program.
Moreover if $\val(B_i)$ is true, the same valuation also
corresponds to a terminating execution starting at the block $Block_i$.
However, it does not mean that there is an
execution that starts in the initial state, visits the block
$Block_i$, and then terminates.  This is because the formula does 
not encode that $Block_i$ is reachable from the initial state.

To overcome this problem one may use the strongest post-condition to
compute the states for which $Block_i$ is reachable. This
roughly doubles the formula.  In our case there is a more simple check
for reachability.  Again, we introduce an auxiliary variable
$\rvar_{i}$ for every block label $i$ that holds
if the execution reaches $Block_i$ from the initial state and
terminates.  Let $Pre_i$ be the set of predecessors of
$Block_i$, i.e., the set of all $j$ such that the final $\goto$
instruction of $Block_j$ may jump to $Block_i$.  Then we can fix the
auxiliary variables $\rvar_i$ using $\WLP$ as follows
\[
\vc: \WLP \land \rvar_0 \equiv B_0 \land 
\bigwedge_{1\leq i\leq n} \biggl(\rvar_i \equiv B_i \land \bigvee_{j\in Pre_i} \rvar_j\biggr).
\]
That is, the reachability variable of the initial block is set to true
if the run is terminating.  The reachability variable of
other blocks is set to true if the current valuation describes a
normally terminating execution starting at this block and at least one
predecessor has its reachability variable set to true.

\begin{theorem}\label{thm:vc}
There is a valuation $\val$ that satisfies $\vc$ with $\val(\rvar_0)=true$
if and only if the 
corresponding initial state leads to a feasible complete path $\pi$ 
for the procedure.
Moreover, the value of the reachability variable $\val(\rvar_i)$ is true 
if and only if there is a path $\pi$ starting in this initial state that visits block $Block_i$.
\end{theorem}
\begin{proof}
  Let there be a feasible path $\pi$ and let $\val$ be the
  corresponding valuation for the initial variables.  If one sets the
  value of each of the auxiliary variable $B_i$ and $\rvar_i$
  according to its definition in $\vc$, then $\vc$ is satisfied by
  $\val$.  Moreover, the $B_i$ variables for every visited block must
  be true by Lemma~\ref{lemma:wlp}.  Then also $\val(\rvar_0)$ must be
  true, i.e., the reachability variable for the initial state must be
  true. By induction one can see that $\val(\rvar_i)$ must also be
  true for every visited block $Block_i$.

  For the other direction, let $\val$ be a satisfying valuation for
  $\vc$ with $\val(\rvar_0)=true$. Then also $\val(B_0)=true$
  holds. Hence, by Lemma~\ref{lemma:wlp} this valuation corresponds to
  a feasible path $\pi$.  Let $\val(\rvar_i)=true$ for some block.  If
  $i=0$ then this is the initial block which is visited by the
  feasible path $\pi$.  For $i\neq 0$ there is 
  some predecessor $j\in Pre_i$ with $\val(\rvar_j) = true$.
  By induction over the
  order of the blocks (note that the code is loop-free)
  one can assume that there is a feasible path starting in this initial state
  that visits $Block_j$.  
  Since $Block_j$ ends with a non-deterministic $\goto$
  that can jump to $Block_i$, the latter block is reachable.  Moreover
  since $\rvar_i$ is true, also $B_i$ must be true and by
  Lemma~\ref{lemma:wlp} the valuation corresponds to a terminating run
  starting at Block $Block_i$.  Thus, there is a run that starts at
  the initial state, reaches block $Block_i$, and terminates.
\end{proof}

Thus $\vc$ is the reachability verification condition that can be used
to generate test cases of the program that reach certain blocks.  To
cover all statements by test cases, one needs to find a set of
valuations for $\vc$, such that each $\rvar_i$ variable is true at least
in one valuation.  The following section will tackle this problem.

\section{Covering algorithms}\label{sec:algorithm}

We can now identify feasible executions through a block simply by
checking if the reachability variable associated with this block evaluates to
$\true$ in a satisfying valuation of the reachability verification condition.
Further, due to the single static assignment performed before generating the formula,
we can identify the initial values for each variable that are needed to
force the execution of this path. That is, a valuation $\val$ satisfying the
$\vc$ can serve as a test case for a block associated with a reachability
variable $\rvar$, if $\val(\rvar)=\true$.
\begin{definition}[Test Case]
Given a reachability verification condition $\vc$ of a program. Let $B$ be a
block in this program, and $\rvar$ be the reachability variable associated with
this block. A \emph{test case} for the block $B$ is a valuation $\val$ of $\vc$,
such that $\val\models\vc$ and $\val(\rvar)$ is $\true$.
\end{definition}

In the following we present two algorithms to compute test cases for loop-free
programs. The first algorithm computes a set of test cases to cover all
feasible control-flow path, the second one computes a more compact set that only
covers all feasible statements.

\paragraph{Path Coverage Algorithm.}
To efficiently generate a set of test cases that covers all feasible
control-flow paths, we need an algorithm that checks which combinations of 
reachability variables in a reachability verification condition can be set to
$\true$. That is, after finding one satisfying valuation for a reachability
verification condition, this algorithm has to modify the next query in a way
that ensures, that the same valuation is not computed again. This procedure has
to be repeated until no further satisfying valuation can be found.

\linesnumbered
\SetAlFnt{\small}
\begin{algorithm}[h]
\KwIn{$\vc$: A reachability verification condition, \\
\quad\quad\quad $\rvars = \{\rvar_0,\dots,\rvar_n\}$: The set of reachability variables}
\KwOut{$\tcs$: A set of test cases covering all feasible paths. }
\SetKwFunction{getVarsWritten}{getVarsWritten}
\SetKwFunction{getVarsRead}{getVarsRead}
\dontprintsemicolon
\Begin{
  $\psi \leftarrow \vc$ \label{src:helper1} \;
  $\tcs \leftarrow \{\}$ \;
  $\val \leftarrow \checksat{\psi}$ \label{src:checksat1}\;
  \While{$\val \neq \{\}$} {
      $\tcs \leftarrow \tcs \cup \{\val \}$ \label{src:addtc}\;
      $\formula \leftarrow \false  $ \;      
	  \ForEach{ $\rvar$ in $\rvars$}{\label{src:buildbc}
	   	\eIf{$\val(\rvar)=\true$} {	   		
	   		$\formula \leftarrow \formula \vee \neg \rvar$ \;	   		
	   	} {	   		
	   		$\formula \leftarrow \formula \vee \rvar$ \;	   		
	   	}
	  }
	  $\psi \leftarrow \psi \wedge \formula $ \label{src:addbc} \;
	  $\val \leftarrow \checksat{\psi}$ \label{src:checksat2} \;
  }
  \Return $\tcs$
}
\caption{\algoA}
\label{alg:maxsat1}
\end{algorithm}

Algorithm $\algoA$, given in Algorithm~\ref{alg:maxsat1}, uses
\emph{blocking clauses} to guarantee that the every valuation is only returned
once. The blocking
clause is the negated conjunction of all assignments to reachability variables
in a valuation $\val$. The algorithm uses the oracle-function \texttt{checksat}
(see line~4 and 16), which has to be provided by
a theorem prover. The function takes a first-order logic formula as input and
returns a satisfying assignment for this formula in form of a set of pair of
variable and value for each free variable in that formula. If the formula is not
satisfiable, \texttt{checksat} returns the empty set. 

The algorithm uses a local copy $\psi$ of the reachability verification condition
$\vc$. As long as \texttt{checksat} is able to compute a satisfying valuation
$\val$ for $\psi$, the algorithm adds this valuation to the set of test
cases $\tcs$ (line~6), and then builds a blocking clause
consisting of the disjunction of the negated reachability variables which are assigned to
$\true$ in $\val$ (line~8). The formula $\psi$ is conjuncted
with this blocking clause (line~13), and the algorithm starts over
by checking if there is a satisfying valuation for the new formula
(line~14). The algorithm terminates when $\psi$ becomes
unsatisfiable. 
\begin{theorem}[Correctness of \algoA]\label{thm:alg1}
Given a loop-free and passive program $P$ with verification condition $\vc$. Let
$\rvars$ be the set of reachability variables used in $\vc$. Algorithm $\algoA$,
started with the arguments $\vc$ and $\rvars$, terminates and returns a set
$\tcs$. 
For any feasible and complete path $\pi$ there is a test case in $\tcs$ for
this path.
\end{theorem}
\begin{proof}
  There are only finitely many solutions for the variables $\rvars$
  that will satisfy the formula $\vc$.  Due to the introduction of the
  blocking clause, every solution will be found only once.  Hence,
  after finitely many iteration the formula $\psi$ must be
  unsatisfiable and the algorithm terminates.  If $\pi$ is a feasible
  and complete path, then by Theorem~\ref{thm:vc} there is a valuation
  $\val$ with $\val(\rvar)=true$ for every block visited by $\pi$.
  Such a valuation must be found by the algorithm before a
  corresponding blocking clause is inserted into $\psi$.  The
  corresponding test case is then inserted into $\tcs$ and is a test
  case for $\pi$.
\end{proof}
Note that $\algoA$ is complete for loop-free programs. For arbitrary programs
that have been transformed using the steps from Section~\ref{sec:loopunwinding},
the algorithm still produces only feasible test cases due to
the soundness of the abstraction.

The advantage of using blocking clauses is that $\algoA$ does not restrict the
oracle \texttt{checksat} in how it should explore the feasible paths encoded in
the reachability verification condition. The drawback of $\algoA$ is that, for each
explored path, a blocking clause is added to the formula and thus, the
increasing size of the formula might slow down the \texttt{checksat} queries if
many paths are explored. This limits the scalability of our algorithm. In
Section~\ref{sec:experiments} we evaluate how the performance of $\algoA$
changes with an increasing size of the input program.

\paragraph{Statement Coverage Algorithm.} 
In some cases one might only be interested in covering all feasible statements.
To avoid exercising all feasible paths, we present a second algorithm, $\algoB$,
in Algorithm~\ref{alg:maxsat2} that computes a compact set of test cases to
cover all feasible statements. The algorithm uses \emph{enabling clauses}
instead of blocking clauses that prevent the oracle from computing the same
valuation twice. An enabling clause is the disjunction of all reachability
variables that have not been assigned to $\true$ by previous satisfying valuation of the
reachability verification condition.

\linesnumbered
\SetAlFnt{\small}
\begin{algorithm}[h]
\KwIn{$\vc$: A reachability verification condition, \\
\quad\quad\quad$\rvars = \{\rvar_0,\dots,\rvar_n\}$: The set of reachability variables}
\KwOut{$\tcs$: A set of test cases covering all feasible statements. }
\SetKwFunction{getVarsWritten}{getVarsWritten}
\SetKwFunction{getVarsRead}{getVarsRead}
\dontprintsemicolon
\Begin{
  $\tcs \leftarrow \{\}$ \;
  $\val \leftarrow \checksat{\vc}$ \;
  \While{$\val \neq \{\}$} {
      $\tcs \leftarrow \tcs \cup \{\val\}$ \label{src2:tacadd}\;
      
	  \ForEach{ $\rvar$ in $\rvars$}{
	   	\If{$\val(\rvar)=\true$} {	   		
	   		$\rvars \leftarrow \rvars \setminus \{\rvar\} $ \label{src2:remove} \;
	   		$\rvars \leftarrow \removedouble{\rvar,\rvars}$ \label{src2:clones} \;	   		
	   	}
	  }	      	 
	   $\formula \leftarrow \false  $ \;
  	   \ForEach{ $\rvar$ in $\rvars$}{\label{src2:ec}
  	   	$\formula \leftarrow \formula \vee \rvar$ \;
  	   }
	  
	  $\val \leftarrow \checksat{\vc \wedge \formula }$ \label{src2:check}\;
  }
  \Return $\tcs$
}
\caption{\algoB}
\label{alg:maxsat2}
\end{algorithm}

The algorithm takes as input a reachability verification condition $\vc$, and the
set of all reachability variables $\rvars$ used in this formula. Like $\algoA$,
$\algoB$ uses the oracle function \texttt{checksat}. First, it checks if there
exists any satisfying valuation $\val$ for $\vc$. If so, the algorithm adds
$\val$ to the set of test cases (line~5). Then, the algorithm
removes all reachability variables from the set $\rvars$, which are assigned to
$\true$ in $\vc$ (line~8). While removing the reachability
variables which are assigned to $\true$, the algorithm also has to check if this
reachability variable corresponds to a block created during loop unwinding. In
that case, all clones of this block are removed from $\rvars$ as well using the
helper function $\removedouble{}$ (line~9). After that, the
algorithm computes a new enabling clause $\formula$ that equals to the disjunction of the
remaining reachability variables in $\rvars$ (line~13) and starts
over by checking if $\vc$ in conjunction with $\formula$ is satisfiable
(line~16). That is, conjunction $\vc\wedge\formula$ restricts the
feasible executions in $\vc$ to those where at least one reachability variable
in $\rvars$ is set to $\true$. Note that, if the set $\rvars$ is empty, the
enabling clause $\formula$ becomes $\false$, and thus the conjunction with $\vc$
becomes unsatisfiable. That is, the algorithm terminates if all blocks have been
visited once, or if there is no feasible execution passing the remaining blocks.
 
\begin{theorem}[Correctness of \algoB]\label{thm:alg2}
Given a loop-free and passive program $P$ with reachability verification condition $\vc$. Let
$\rvars$ be the set of reachability variables used in $\vc$. Algorithm $\algoB$,
started with the arguments $\vc$ and $\rvars$, terminates and returns a set
$\tcs$.  For any block in the program
there exists a feasible paths $\pi$ passing this block if and only if
there exists a test case $\val\in\tcs$, that passes this block.
\end{theorem}
\begin{proof}
  In every iteration of the loop at least one variable of the set
  $\rvars$ will be removed.  This is because the formula $\phi$ will
  only allow valuations such that for at least one $\rvar\in\rvars$
  the valuation $\val(\rvar)$ is true.  Since $\rvars$ contains only
  finitely many variables the algorithm must terminate.  If $\pi$ is a
  feasible path visiting the block associated with the variable
  $\rvar$, then there is a valuation $\val$ that satisfies $\vc$ with
  $\val(\rvar) = true$.  Such a valuation must eventually be found, since
  $\vc\land\phi$ is only unsatisfiable if $\rvar \notin\rvars$.  The
  valuation is added to the set of test cases $\tcs$.
\end{proof}

The benefit of $\algoB$ compared to $\algoA$ is that it will produce at most
$|\rvars|$ test cases, as each iteration of the loop will generate only one test
case and remove at least one element from $\rvars$. That is, the resulting set
$\tcs$ can be used more efficiently if only statement coverage is needed.
However, the enabling clause might cause the theorem prover which realizes
\texttt{checksat} to take detours or throw away information which could be
reused. It is not obvious which of both algorithms will perform better in terms
of computation time. Therefore, in the following, we carry out some experiments
to evaluate the performance of both algorithms.

Note that, like $\algoA$, $\algoB$ is complete for loop-free programs and sound
for arbitrary programs. That is, any block that is not covered by these
algorithms is unreachable code (in the loop-free program).

\section{Procedure Summaries}\label{sec:procsum}

For large programs, inlining all procedure calls as proposed in
Section~\ref{sec:loopunwinding} might not be feasible. However, replacing them
by using assume-guarantee reasoning as it is done, e.g., in static checking~\cite{Barnett06boogie} is not a feasible solution either. Using
contracts requires the necessary expertise from the programmer to write
proper pre- and postconditions, and thus, it would violate our goal of
having a fully automatic tool. If trivial contracts are generated automatically
(e.g., \cite{doomedjournal}), it will introduce feasible executions that do not
exist in the original program. This would break the \emph{soundness} requirement
from Lemma~\ref{thrm:soundness} that each of the test cases returned by the
algorithms $\algoA$ and $\algoB$ must represent a feasible path in the
(loop-free) program.

Instead of inlining each procedure call, we propose to replace them by a
\emph{summary} of the original procedure which represents \emph{some} feasible
executions of the procedure. The summary can be obtained directly by applying
$\algoA$ or $\algoB$ to the body of the called procedure. Each valuation $\val$
in the set $\tcs$ returned by these algorithms contains values for all
incarnations of the variables used in the procedure body on one feasible
execution. In particular, for a variable $v$, with the first incarnation $v_0$
and the last incarnation $v_n$, $\val(v_0)$ represents one feasible input
value for the considered procedure and $\val(v_n)$ represents the value of $v$
after this procedure returns. That is, given a procedure $P$ with
verification
condition $\vc$ and reachability variables $\rvars$, let $\tcs
=
\algoA(\vc,\rvars)$ or $\tcs = \algoB(\vc,\rvars)$
respectively. Furthermore let $V$ be the
set of variables which are visible to the outside of $P$, that is, parameters
and global variables. The summary $Sum$ of $P$ is expressed by the formula:
\[
Sum := \bigvee_{\val\in\tcs} 
(
\bigwedge_{v\in V}(v_0=\val(v_0))
\wedge
\bigwedge_{v\in V}(v_n=\val(v_n)) 
),
\]
where $n$ refers to the maximum incarnation of a particular variable $v$. The
summary can be interpreted as encoding each feasible path of $P$ by the
condition that, if the initial values for each variable are set appropriately,
the post-state of this execution is established. 
We need an underapproximation of the feasible executions of the procedure
as the procedure summary. Therefore we encode the summary of the previously 
computed
paths and let the theorem prover choose the right
path. In practice, in
particular when using $\algoA$, it can be useful to consider only a subset of
$\tcs$ for the summary construction, as a formula representing all paths might
outgrow the actual verification condition of the procedure.

On the caller side, we can now replace the call to a procedure $P$ by an
assumption $\assume Sum$ where $Sum$ is the procedure summary of $P$. We further have to add some framing assignments to
map the input- and output variables of the called procedure to the one of the
calling procedure. We illustrate this step using the following example program:
\begin{verbatim}
proc foo(a, b) returns c {
  l1: 
      goto l2, l3;
  l2: assume b > 0;
      c := a + 1;
      goto l4;
  l3: assume b <= 0;
      c := a - 1;
      goto l4;
  l4:
}

proc bar(x) returns z {
  l1:       
      z := call foo(x,1);      
}
\end{verbatim}
Applying the algorithm $\algoA$ to the procedure \verb|foo| will result in a
summary like:
\[
Sum := 
 \begin{array}[c]{r}
 (a_0=0\wedge b_0=0)\wedge(c_1=a_0-1) \\
 \vee (a_0=0\wedge b_0=1)\wedge(c_1=a_0+1)
 \end{array}
\]
This summary can be used to replace the call statement in \verb|bar| after the
single static assignment has been performed as follows:
\begin{verbatim}
proc bar(x) returns z {
  l1:       
      assume a0=x;
      assume b0=1;
      assume Sum;
      assume z1=c1;   
}
\end{verbatim}
Note that, to avoid recomputing the single static assignment, when reaching a
call statement, we increment the incarnation count for each variable that might
be modified by this procedure and the incarnation count of each global variable.
Therefore, we have to add frame conditions if a global variable is not changed
by the summary (in this example it is not necessary, as there are no global
variables).

A procedure summary can be seen as a switch case over possible input values.
That is, the summary provides the return values for a particular set of input
values to the called procedure. Any execution that calls the procedure with
other input values becomes infeasible. In that sense, using procedure summaries
is an under-approximation of the set of feasible executions and thus sound for
our purpose.
\begin{lemma}[Soundness]\label{lm:inlinesound}
Given a loop-free procedure $P$ which calls another loop-free procedure $P'$.
Let $P^\#$ be the version of procedure $P$ where all calls to $P'$ have been
replaced by the summary of $P'$. Any feasible execution of $P^\#$ is also a
feasible execution of $P$.
\end{lemma}

Using these summaries is a very strong abstraction as only a very
limited number of possible input values is considered as the set of feasible executions of the
called procedure is reduced to one per control-flow path (or even less, if
algorithm $\algoB$ is used). In particular, this causes problems if a procedure
is called with constant values as arguments. In the example above, inlining only
works if the theorem prover picks the same constant when computing the summary
that is used on the caller side (which is the case here). If the constants do
not match, the summary might provide no feasible path through the procedure,
which is still sound but not useful. In that case, a new summary has to be
computed where the constant values from the caller side are used as a
precondition for the procedure (e.g., by adding an appropriate assume statement
to the first block of the called procedure) before re-applying algorithm
$\algoA$ or $\algoB$.

The benefit of this summary computation is that it is fully automatic and the
computation of the summary is relatively cheap, because the called procedure has
to be analyzed at least once anyway. However, it is not a silver bullet and
its practical value has to be evaluated in our future work. We do not consider
procedure summaries as an efficient optimization. They rather are a necessary
abstraction to keep our method scalable.
  
\section{Experiments}\label{sec:experiments}

We have implemented a prototype of the presented algorithms. As this prototype
still is in a very early stage of development, the goal of this experiments is
only to evaluate the computation time of the queries needed to cover all
feasible statements in comparison to similar approaches. Other experiments, such
as the applicability to real world software remain part of future work.

We compare the algorithms from Section~\ref{sec:algorithm} with two other
approaches that compute a covering set of feasible executions: A worst-case
optimal approach $\algoFM$ from \cite{Hoenicke:2009:DWP:1693345.1693374} and a
query-optimal approach $\algoVSTTE$ from \cite{vstte12}. The worst-case optimal
approach checks if there exists a feasible control-flow path passing each
\emph{minimal} block. A block is minimal, if there exists no block that is
passed by a strict subset of the executions passing this
block~\cite{doomedjournal,Bertolino:1993:UEA:156928.156932}. Each implementation
uses helper variables to build queries that ask the theorem prover for the
existence of a path passing through one particular block. The query-optimal
approach~\cite{vstte12} uses helper variables to count how many minimal elements
occur on one feasible execution and then applies a greedy strategy to cover as
many minimal elements as possible with one valuation of the formula. Note that
the purpose of $\algoFM$ and $\algoVSTTE$ is slightly different from the purpose
of the algorithms in this paper. The $\algoFM$ and $\algoVSTTE$ use a loop-free
abstraction of the input program that over-approximates the set of feasible
executions of the original program (see, \cite{doomedjournal}). On this
abstraction they prove the existence of blocks which cannot be part of any
terminating execution. To be comparable, we use the same abstraction for
all algorithms. That is, we use $\algoFM$ and $\algoVSTTE$ to check the
existence of statements that do not occur on feasible executions. Since both
algorithms are complete for loop-free programs, they return the same result as
$\algoB$.

Note that the result and purpose of all algorithms is slightly different.
However, all of them use a theorem prover as an oracle to identify executions
that cover all feasible statements in a program.

For now, our implementation works only for the simple language from
Section~\ref{sec:preliminaries}. An implementation for a high-level language is
not yet available. Hence the purpose of the experiments is only to measure the
efficiency of the queries. Therefore, we decide to use randomly generated
programs as input data. Generated programs have several benefits. We can control
the size and shape of the program, we can generate an arbitrary number of
different programs that share some property (e.g., number of control-flow
diamonds), and they often have lots of infeasible control-flow paths. We are
aware that randomly generated input is a controversial issue when evaluating
research results, but we believe that, as we want to evaluate the performance of
the algorithms, and not their detection rate or practical use, they are a good
choice. A more technical discussion on this issue follows in the threats to
validity.

\paragraph{Experimental Setup.}

As experimental data, we use 80 randomly generate unstructured programs. Each
program has between 2 and 9 control-flow diamond shapes, and each diamond shape
has 2 levels of nested if-then-else blocks (i.e., there are 4 distinct paths
through each diamond). A block has 3 statements, which are either assignments of
(linear arithmetic) expressions to unbounded integer variables or assumptions
guarding the conditional choice. Each program has between 90 and 350 lines of
code and modifies between 10 and 20 different variables.

For each number of control-flow diamonds, we generated 10 different random
programs and computed the average run-time of the algorithms. This is necessary to get an
estimate of the performance of each algorithm, as their computation time
strongly depends on the overall number of feasible executions in the analyzed
program.

For a fair comparison, we use the theorem prover
SMTInterpol~\footnote{\url{http://ultimate.informatik.uni-freiburg.de/smtinterpol/}}
in all four algorithms. For each algorithm, we record how often the theorem
prover is asked to check the satisfiability of a formula and we record the time
it takes until the theorem prover returns with a result. All experiments are
carried out on a standard desktop computer with ample amount of memory.

\begin{table}[t]
\begin{center}
\begin{tabular}{ |l|r|r| }
\hline 
	Algorithm & Queries (total) & Time (sec) \\
  \hline
\algoA   & 69777 & 13.12 \\  
\algoB & 854 & 11.81 \\
\algoFM  & 1760 & 145,45 \\
\algoVSTTE  & 512 & 615,51 \\
\hline
\end{tabular}
\end{center}
\caption{Comparison of the four algorithms in terms of total number of queries
and computation time for 80 benchmark programs.}
\label{tbl:results}
\captionspace
\end{table}

\paragraph{Discussion.} Table~\ref{tbl:results} shows the summary of the results
for all algorithms after analyzing 80 benchmark programs.
Figure~\ref{fig:time} gives a more detailed view on the computation time per
program. The x-axis scales over the number of control-flow diamonds ranging from
2 diamonds to 9. Figure~\ref{fig:queries} gives a detailed view on the number
of queries. As before, the x-axis scales over the number of control-flow
diamonds.

The algorithms $\algoA$ and $\algoB$ are clearly faster than $\algoFM$ and
$\algoVSTTE$. Overall, $\algoB$ tends to be the fastest one.
Figure~\ref{fig:compare} shows the computation time for $\algoA$ and $\algoB$ in
a higher resolution. \begin{figure}[h] \centering 
\epsfig{file=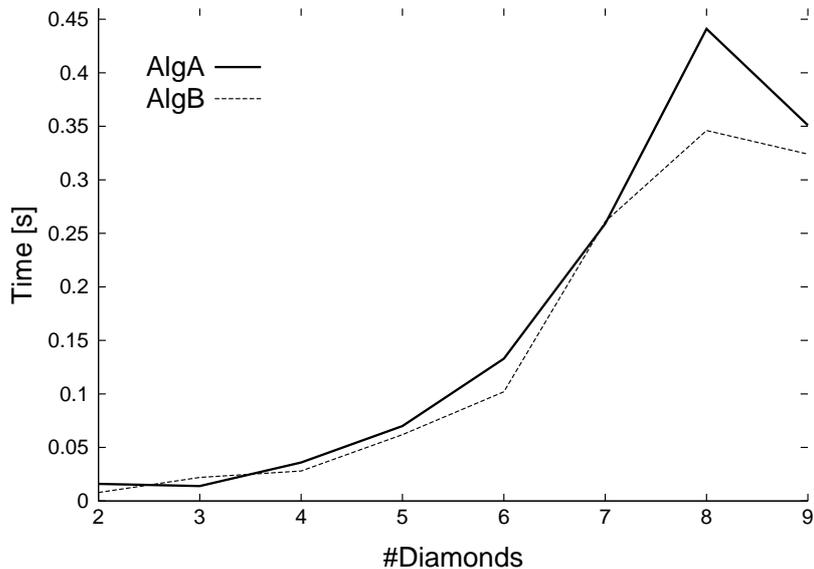,width=0.90\linewidth,clip=}
\caption{Runtime comparison of the algorithms proposed in Section~\ref{sec:algorithm}. The ticks on the x-axis represent
the number of control-flow diamonds in the randomly generated programs.
\label{fig:compare}}
\captionspace
\end{figure}  
It turns out that the difference between the computation time of $\algoA$ and
$\algoB$ tends to become bigger for larger programs. As expected, $\algoB$ works
a bit more efficient as the size of the formula is always bounded, while
$\algoA$ asserts one new term for every counterexample found. However, comparing
the number of theorem prover calls, there is a huge difference between $\algoA$
and $\algoB$. While $\algoB$ never exceeds a total of 20 queries per
program, $\algoA$ skyrockets already for small programs. For a program with 10
control-flow diamonds, $\algoA$ uses more than 2000 theorem prover calls, where
$\algoB$ only need 10. Still, $\algoB$ is only $0.03$ seconds faster on this
example ($<10\%$). 

\begin{figure}[h] \centering
\epsfig{file=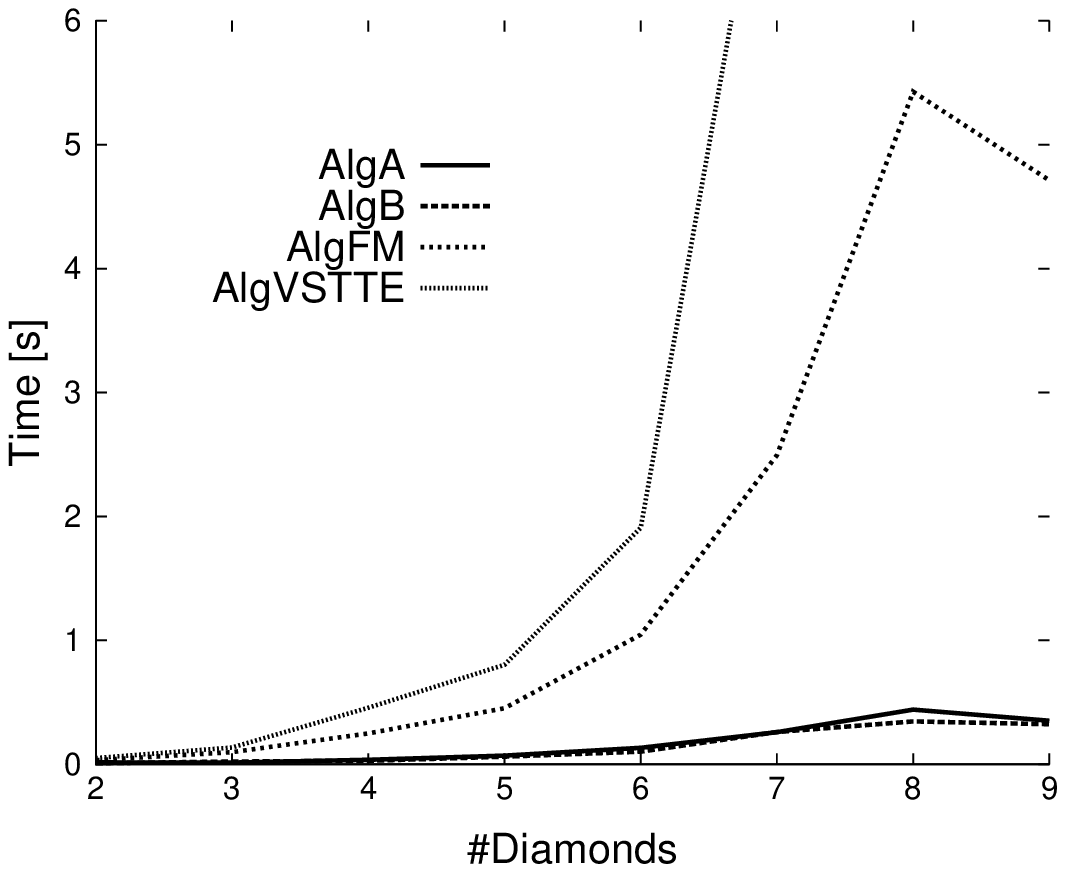,width=.95\linewidth,clip=}
\caption{Computation time for each algorithm. The ticks on the x-axis represent
the number of control-flow diamonds in the randomly generated programs.
\label{fig:time}}
\captionspace
\end{figure}  

\begin{figure}[h]
\centering
\epsfig{file=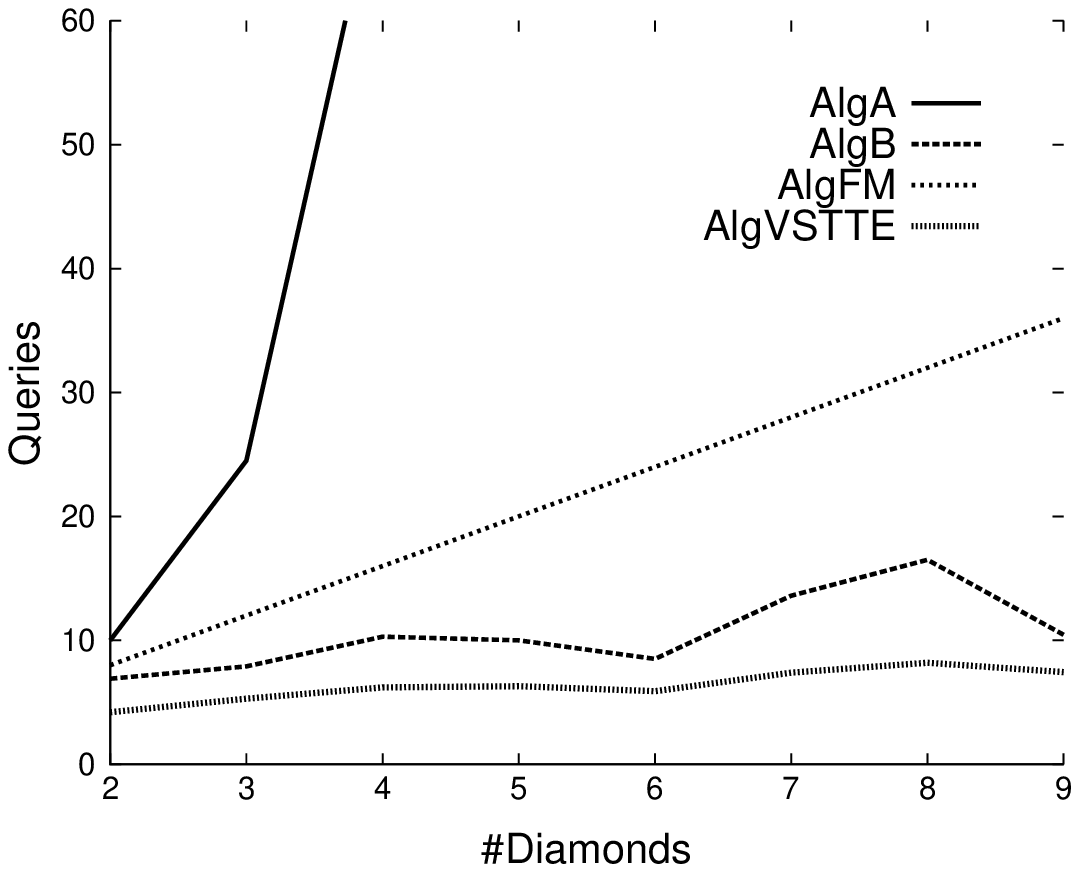,width=.95\linewidth,clip=}
\caption{Number of call to the theorem prover for each algorithm. The ticks on
the x-axis represent the number of control-flow diamonds in the randomly generated programs. \label{fig:queries}}
\captionspace
\end{figure} 

These results show that, even though $\algoB$ might be slightly more efficient
than $\algoA$, the number of queries is not an important factor for the
computation time. In fact, internally, the theorem prover tries to find a new
counterexample by changing as few variables as possible which is very close to
the idea of $\algoA$. $\algoB$, which queries if there exists a counterexample
through a block that has not been visited so far, will internally perform the
same steps as $\algoA$ and thus, the performance gain is only rooted in the
smaller formulas and reduced communication between application and prover.
However, the results also show that, when using a theorem prover, computing a
path cover with $\algoA$ is not significantly more expensive than computing
only a statement cover with $\algoB$.

The computation time for $\algoFM$ and $\algoVSTTE$ are significantly higher
than the one for the presented algorithms. For $\algoFM$, some queries,
and thus some computation time, could be saved by utilizing the counterexamples
to avoid redundant queries. However, the number of queries cannot become better
than the one of $\algoB$ due to the kind of queries. The most significant
benefit of $\algoA$ and $\algoB$ over $\algoFM$ is that they don't have to
inject helper variables in the program. In fact $\algoA$ and $\algoB$ also use
one variable per block to encode the reachability, but this variable is added to
the formula and not to the program. Thus, it is not considered during single
static assignment, which would create multiple copies for each variable.

For the query-optimal algorithm $\algoVSTTE$, the computation time becomes
extremely large for our random programs. This is due to the fact that
$\algoVSTTE$ tries to find the best possible counterexample (that is, the one
with the most previously uncovered blocks) with each query. Internally, the
theorem prover will exercise several counterexamples and discard them
until the best one is found. The procedure is similar to the one used in
$\algoA$ and $\algoB$: the theorem prover computes a counterexample and then
assures that this example cannot be found again, and then starts over. But in
contrast to $\algoVSTTE$, our algorithms do not force the theorem prover to
find a path that satisfies additional constraints, and, hence, relaxing the
problem that has to be solved by the theorem prover.
Even though one might find benchmarks where
$\algoVSTTE$ is significantly faster than $\algoFM$, the algorithms $\algoA$ and
$\algoB$ will always be more efficient since they pose easier (and, hence,
faster) queries to the theorem prover.

The presented results should not be interpreted as an argument against a
query-optimal algorithm. We rather conclude that the place for such
optimizations is inside the theorem prover. Modifying the way, the theorem
prover finds a new counterexample can lead to tremendous performance
improvements. However, such changes have to consider the structure of
verification conditions and thus will exceed the functionality of a general
theorem prover.

\paragraph{Threats to validity.} 
We emphasize that the purpose of the experiments is only to evaluate
the performance of $\algoA$ and $\algoB$. These experiments are not valid to
reason about practical use or scalability of the method.

We report several internal threats to validity:
The experiments only used a very restricted background theory.
However, the path reasoning described in this paper prunes the search space for the theorem prover
even if we use richer logics including arrays or quantifiers. As shown in our
experiments, the algorithms proposed in this paper pose easier problems to a
theorem prover. This won't change if we switch to richer logics since our
algorithms only limit the theorem prover to reason about feasible paths while all
other algorithms pose additional constraints on such a path. If we use richer
logics we only limit the number of paths. But still it remains easier to
just find a path than to find one that satisfies some
additional condition.

We have chosen randomly generated programs as input for two reasons. First, we
wanted to be able to scale the number of paths and use the most difficult shape
of the control structure for our techniques. Hence, we had to scale the number
of diamonds in the control flow graph. Second, we did not implement a parser for
a specific language. Existing translations from high-level languages into
unstructured languages are not suitable for our algorithms as they
over-approximate the set of infeasible executions to retain soundness w.r.t.
partial correctness proofs. These translations might both over- and
under-approximate the set of feasible executions of a program and thus violate
our notion of soundness. However, for the purpose of comparing the performance
of the different algorithms, the experiments are still valid.

In our experiments we only used SMTinterpol to answer the queries. For the
comparison of $\algoA$ with the other algorithms, the choice of the
theorem prover can make a significant difference. SMTinterpol tries to find a
valuation for a formula by making as few changes as possible to the previous
valuation. If a theorem prover chooses a different strategy, in particular
$\algoB$ might become much fast. However, we are not aware of any theorem prover
that uses this kind of strategy.

\section{Related Work}\label{sec:related}
Automatic test case generation is a wide field ranging from purely random
generation of input values (e.g., \cite{PachecoLET2007}) to complex static
analysis. The presented algorithms can best be compared to tools that provide
automatic white-box test case generation. Probably the most notable tools in
this field are PREfix~\cite{Bush:2000:SAF:348422.348428} and
Pex~\cite{DBLP:conf/sigsoft/TillmannS05a}. Both algorithms use symbolic
execution to generate test cases that provoke a particular behavior. Pex further
allows the specification of parameterized unit tests. Symbolic execution analyzes
a program path-by-path and then uses constraint solving to identify adequate
input to execute this path. In contrast, our approach encodes all paths into one
first-order formula and then calls a theorem prover to return any path and the
input values needed to execute this path. In a way, symbolic execution selects
a path and then searches feasible input values for this path, while our approach
just asks the theorem prover for \emph{any} path which is feasible. One
advantage of our approach is that it might be more efficient to ask the theorem
prover for a feasible path than checking for each path if it is feasible.

Many other approaches to static analysis-based automatic test case generation
and bounded model checking exist but, due to the early stage of the development
of the proposed ideas, a detailed comparison is subject to future work.

In \cite{Engel07generatingunit} test cases are generated from interactive
correctness proofs. The approach of using techniques from verification to
identify feasible control-flow paths for test case generation is similar to
ours. However, they generate test cases from a correctness proof, which
might contain an over-approximation of the feasible executions. This can
result in non-executable test cases. Our approach under-approximates the
set of feasible executions and thus, any of the generated test cases can be
executed.

Using a first-order formula representation of a program and a theorem prover to
identify particular paths in that program goes back to, e.g.,
ESC~\cite{Flanagan:2002:ESC:543552.512558} and, more recently,
Boogie~\cite{Leino:2005:EWP:1066417.1066421,Flanagan:2001:AEE:360204.360220,Barnett06boogie}.
These approaches use similar program transformation steps to generate the
formula representation of a program. However, the purpose of these approaches is
to show the absence of failing executions. Therefore, their formula represents
an approximation of the weakest precondition of the program with postcondition
$\true$. In contrast, we use the negated $\wlp$ with postcondition $\false$.
Showing the absence of failing executions is a more complicated task and
requires a user-provided specification of the intended behavior of the program.

In \cite{Grigore:2009:SPU:1557898.1557904}, Grigore et al propose to use the
strongest postcondition instead of the weakest precondition. This would also be
possible for our approach. As mentioned in Section~\ref{sec:vc}, the
reachability variables are used to avoid encoding the complete strongest
postcondition. However, it would be possible to use $sp$ and modify the
reachability variables to encode $\wlp$.

Recently there has been some research on $\wlp$ based program analysis: in
\cite{1292319}, an algorithm to detect unreachable code is presented. This
algorithm can be seen as a variation of $\algoB$. However, it does not return
test cases. The algorithms
$\algoFM$~\cite{Hoenicke:2009:DWP:1693345.1693374,doomedjournal}, and
$\algoVSTTE$~\cite{vstte12} detect code which never occurs on feasible
executions. While $\algoFM$ detects \emph{doomed program points}, i.e.
control-locations, $\algoVSTTE$ detects statements, i.e. edges in the CFG.
If a piece of code cannot be proved doomed/infeasible, a counter example is
obtained which represents a normal-terminating executions. The main difference
to our approach is that their formula is satisfied by all executions that either
block \emph{or fail}. We do not consider that an execution might fail and leave
this to the execution of the test case.

There are several strategies to cover control-flow graphs. The most related to
this work is \cite{vstte12}, which has already been explained above. Other
algorithms such as,
\cite{Bertolino:1994:AGP:203102.203103,Bertolino:1993:UEA:156928.156932,Forgacs:1997:FTP:267896.267922}
present strategies to compute feasible path covers efficiently. These algorithms
use dynamic analysis and are therefore not complete.

Lahiri et al~\cite{DBLP:conf/cav/LahiriNO06} used a procedure similar to one
of our proposed algorithms to efficiently compute predicate abstraction.
They used an AllSMT loop over a set of \textsl{important} predicates.
One of our algorithms, $\algoA$, lifts this idea to the context of test case generation
and path coverage.
Our second algorithm, $\algoB$ cannot be used in their context since the authors of this
paper need to get all satisfying assignments for the set of predicates.
In contrast, we are only interested in the set of predicates that are
satisfied in at least one model of the SMT solver.

\section{Conclusion}\label{sec:conclusion}

We have presented two algorithms to compute test cases that cover those
statements respectively control-flow paths which have feasible executions within
a certain number of loop unwindings. The algorithms compute a set of test cases in a 
fully automatic way without requiring any user-provided information about the
environment of the program. The algorithms guarantee
that these executions also exist in the original program (with loops). We
further have presented a fully automatic way to compute procedure summaries,
which gives our algorithm the potential to scale even to larger programs.

If no procedure summaries are used, the presented algorithms  cover \emph{all}
statements/paths with feasible executions within the selected number of
unwindings. That is, besides returning test cases for the feasible
statements/paths one major result is that all statements that are not covered
cannot be covered by \emph{any} execution and thus are dead code.

The experiments show that the preliminary implementation already is able to
outperform existing approaches that perform similar tasks. The experiments also
show that computing a feasible path cover is almost as efficient as computing a
feasible statement cover with the used oracle even for procedures of up to 300
lines of code.

Due to the early stage of development there are still some limitations which
refrain us from reporting a practical use of the proposed algorithms. So far, we
do not have a proper translation from high level programming language into our
intermediate format. Current translations into unstructured intermediate
verification languages such as Boogie~\cite{Barnett06boogie} are built to preserve all
failing executions of a program for the purpose of proving partial correctness.
However, these translations add feasible executions to the program during
translation which breaks our notion of soundness. Further, our language does not
support assertions. Runtime errors are guarded using conditional choice to give
the test case generation the possibility to generate test cases that provoke
runtime errors. A reasonable translation which only under-approximates feasible
executions is still part of our future research.

Another problem is our oracle. Theorem provers are limited in their ability to
find satisfying valuations for verification conditions. If the program contains,
e.g., non-linear arithmetic, a theorem prover will not be able to find a
valuation in every case. This does not affect the soundness of our approach, but
it will prevent the algorithm from covering all feasible paths (i.e., the approach is
not complete anymore). To make these algorithms applicable to real world
programs, a combination with dynamic analysis might be required to identify
feasible executions for those parts where the code is not available, or where
the theorem prover is inconclusive.

\paragraph{Future Work.} Our future work encompasses the development of a proper
translation from Java into our unstructured language. This step is essential to
evaluate the practical use of the proposed method and to extend its use to other
applications.

One problem when analyzing real programs is intellectual property boundaries and
the availability of code of third-party libraries. We plan to develop a
combination of this approach with random testing (e.g.,
\cite{DBLP:conf/oopsla/PachecoE07}), where random testing is used to compute
procedure summaries for library procedure(s) where we cannot access the code.

The proposed procedure summaries have to be recomputed 
if the available summaries for a procedure do not represent any feasible
execution in the current calling context. 
Therefore, we plan to develop a refinement loop
which stores summaries more efficiently.

Another application would be to change the reachability variables in a way that
they are only true if an assertion inside a block fails rather than if the block
is reached. This would allow us to identify \emph{all} paths that violate
assertions in the loop-free program. Encoding failing assertions this way can be
seen as an extension of the work of Leino et al in
\cite{Leino:2005:GET:1065095.1065103}.

In the theorem prover, further optimizations could be made to improve the
performance of $\algoB$. Implementing a strategy to find new valuations that,
e.g., change as many reachability variables as possible from the last valuation
could lead to a much faster computation of a feasible statement cover. In the
future, we plan to implement a variation of the algorithm
$\algoVSTTE$~\cite{vstte12} inside the theorem prover.

We believe that the presented method can be a powerful extension to dynamic
program analysis by providing information about which parts of a program can be
executed within the given unwinding, what valuation is needed to execute them,
and which parts can never be executed. The major benefit of this kind of program
analysis is that it is \emph{user friendly} in a way that it does not require
any input besides the program and that any output refers to a real execution in
the program. That is, it can be used without any extra work and without any
expert knowledge. However, more work is required to find practical evidence for
the usefulness of the presented ideas.

\paragraph{\textbf{Acknowledgements.}} This work is supported
  by the projects ARV and COLAB funded by Macau Science and Technology
  Development Fund.

\bibliographystyle{abbrv}
\bibliography{main}

\end{document}